\documentclass[10pt,conference,romanappendices]{IEEEtran}
\usepackage[cmex10]{amsmath}
\usepackage{amssymb}
\usepackage{amsfonts}
\usepackage{wrapfig}
\usepackage{graphicx}
\usepackage{enumerate}
\usepackage{xspace}

\newcommand{\cs}{\mbox{$\mi{CmplSet}$}\xspace}

\newtheorem{example}{Example}
\newtheorem{definition}{Definition}
\newtheorem{proposition}{Proposition}
\newenvironment{proof}{\hspace{8pt}\ti{Proof:}}{}

\newtheorem{remark}{Remark}

\newcommand{\bm}[1]{{\mbox{\boldmath $#1$}}}

\newcommand{\di}[1]{\mbox{$\mi{Diam}(#1)$}\xspace}
\newcommand{\rch}[1]{\mbox{$\mi{Rch}(#1)$}\xspace}

\newcommand{\imp}{\Rightarrow}

\newcommand{\pnt}[1]{{\mbox{$\vec{#1}$}}}
\newcommand{\Pnt}[1]{{\mbox{$\vec{#1}\,'$}}}
\newcommand{\ppnt}[2]{{\mbox{$\vec{#1}_{#2}$}}}

\newcommand{\V}[1]{\mbox{$\mathit{Vars}(#1)$}}

\newcommand{\s}[1]{\mbox{$\{#1\}$}}

\newcommand{\nGz}[2]{$G_{non-\{z\}}$}

\newcommand{\prr}[1]{\mi{Prev}(\boldsymbol{q})}

\newcommand{\mi}[1]{\mathit{#1}}
\newcommand{\ti}[1]{\textit{#1}}
\newcommand{\tb}[1]{\textbf{#1}}

\newcommand{\ttt}{\>\>\>}

\newcommand{\Tt}{\>\>}
\newcommand{\tbl}{\mbox{$T_{\mi{tbl}}$}\xspace}
\newcommand{\stbl}{\mbox{$T^*_{\mi{tbl}}$}\xspace}

\newcommand{\prob}[2]{\mbox{$\exists{#1} [#2]$}}
\newcommand{\Prob}[3]{\mbox{$\exists{#1}\exists{#2}[#3]$}}

\newcommand{\Comment}[1]{}

\newcommand{\abs}[1]{\mbox{$\mathcal{#1}$}}
\newcommand{\Abs}[2]{\mbox{$\mathcal{#1}^{\mi{#2}}$}}

\begin{document}

\title{On Verifying Designs With Incomplete Specification}

\author{\IEEEauthorblockN{Eugene Goldberg} 
\IEEEauthorblockA{
eu.goldberg@gmail.com}}

\maketitle
\begin{abstract}
Incompleteness of a specification \ti{Spec} creates two problems.
First, an implementation \ti{Impl} of \ti{Spec} may have some
\ti{unwanted} properties that \ti{Spec} does not forbid. Second,
\ti{Impl} may break some \ti{desired} properties that are not in
\ti{Spec}.  In either case, \ti{Spec} fails to expose bugs of
\ti{Impl}.  In an earlier paper, we addressed the first problem above
by a technique called Partial Quantifier Elimination (PQE).  In
contrast to complete QE, in PQE, one takes out of the scope of
quantifiers only a small piece of the formula. We used PQE to generate
properties of \ti{Impl} i.e. those \ti{consistent} with
\ti{Impl}. Generation of an unwanted property means that \ti{Impl} is
buggy.  In this paper, we address the second problem above by using
PQE to generate false properties i.e those that are \ti{inconsistent}
with \ti{Impl}.  Such properties are meant to imitate the missing
properties of \ti{Spec} that are not satisfied by \ti{Impl} (if
any). A false property is generated by modifying a piece of a
quantified formula describing 'the truth table' of \ti{Impl} and
taking this piece out of the scope of quantifiers.  By modifying
different pieces of this formula one can generate a ``structurally
complete'' set of false properties.  By generating tests detecting
false properties of \ti{Impl} one produces a high quality test set. We
apply our approach to verification of combinational and sequential
circuits.
\end{abstract}

\section{Introduction}
One of the drawbacks of formal verification is that the set of
properties describing a design usually does not specify the latter
completely. Incompleteness of a specification \ti{Spec} creates two
problems.  First, an implementation \ti{Impl} of \ti{Spec} may have
some \ti{unwanted} properties that \ti{Spec} does not forbid. Second,
\ti{Impl} may break some \ti{desired} properties that are not in
\ti{Spec}.  In either case, \ti{Spec} fails to expose bugs of
\ti{Impl}. In testing, the incompleteness of verification is addressed
by using a set of tests that is complete \ti{structurally} rather than
functionally.  Structural completeness is achieved by probing every
piece of the design under test.

In~\cite{con_props}, we used the idea of structural completeness to
attack the first problem above. The idea was to use a technique called
partial quantifier elimination (PQE) to generation properties
\ti{consistent} with \ti{Impl}. In contrast to complete quantifier
elimination, PQE takes out of the scope of quantifiers only a small
piece of formula.  A property $Q(V)$ of \ti{Impl} is produced by
applying PQE to formula \prob{W}{F(V,W)} defining ``the truth table''
of \ti{Impl}.  Here $F$ describes the functionality of \ti{Impl} and
$V,W$ are sets of external and internal variables of \ti{Impl}
respectively. If $Q$ is not implied by \ti{Spec}, then the latter is
incomplete.  If $Q$ describes an unwanted property of \ti{Impl}, the
latter is buggy.  Otherwise, a new property is added to \ti{Spec} to
make it imply $Q$.  By taking different pieces of $F$ out of the scope
of quantifiers, one can build a specification that is \ti{structurally
  complete}.

In this paper, we continue this line of research by addressing the
second problem above. Namely, we use PQE to generate false properties
i.e. those \ti{inconsistent} with \ti{Impl}. They are meant to imitate
the missing properties of \ti{Spec} not satisfied by \ti{Impl} (if
any). Tests breaking a false property may expose a bug that was not
discovered due to \ti{Spec}'s lacking a property not satisfied by
\ti{Impl}. A false property $Q(V)$ is generated in two steps. First, a
new formula $F^*$ is obtained from $F$ by a slight modification. Then
the modified part is taken out of the scope quantifiers from
\prob{W}{F^*}. If $F^**$ and $F$ are not logically equivalent, this
produces a property $Q$ that is implied by $F^*$ but not by $F$. By
modifying different parts of $F$ one can generate a ``structurally
complete'' set of false properties. By generating tests breaking these
properties one can build a high quality test.

Our contribution is as follows. First, we show that one can use PQE to
generate false properties of \ti{Impl} for combinational
circuits. Second, we describe an algorithm that, given a combinational
circuit, forms a structurally complete set of false properties of this
circuit. Third, we extend our approach to sequential circuits.
Fourth, to show the high quality of tests generated off false
properties we relate the former to tests detecting stuck-at faults.

The main body\footnote{Some additional information is given in the
  appendix.} of this paper is structured as follows. Basic definitions
are given in Section~\ref{sec:basic}. In
Section~\ref{sec:gen_fls_props}, we describe generation of false
properties for combinational circuits.  A procedure for building a
structurally-complete set of false properties is given in
Section~\ref{sec:cmpl_set}. Section~\ref{sec:seq_circs} extends our
approach to sequential circuits by showing how one can generate false
safety properties. We relate our approach to fault/mutation detection
in Section~\ref{sec:relation}. Finally, we make some conclusions in
Section~\ref{sec:concl}.

\section{Basic Definitions}
\label{sec:basic}
%
%
In this paper, we consider only propositional formulas. We assume that
every formula is in conjunctive normal form (CNF).  A \ti{clause} is a
disjunction of literals (where a literal of a Boolean variable $w$ is
either $w$ itself or its negation $\overline{w}$). So a CNF formula
$H$ is a conjunction of clauses: $C_1 \wedge \dots \wedge C_k$. We
also consider $H$ as the \ti{set of clauses} \s{C_1,\dots,C_k}.

\begin{definition}
Let $V$ be a set of variables. An \tb{assignment} \pnt{q} to $V$ is a
mapping $V~\rightarrow \s{0,1}$.
\end{definition}

%
%
\begin{definition}
  \label{def:vars}
Let $H$ be a formula. \bm{\V{H}} denotes the set of variables of $H$.
\end{definition}

%
%
\begin{definition}
  \label{def:qe_prob}
Let $H(V,W)$ be a formula where $V,W$ are sets of variables.  The
\tb{Quantifier Elimination (QE)} problem specified by \prob{V}{H} is
to find formula $H^*(W)$ such that \bm{H^* \equiv \prob{V}{H}}.
\end{definition}

%
%
\begin{definition}
  \label{def:pqe_prob}
Let $H_1(V,W)$, $H_2(V,W)$ be formulas where $V,W$ are sets of
variables.  The \tb{Partial QE} (\tb{PQE}) problem of taking $H_1$ out
of the scope of quantifiers in \prob{V}{H_1 \wedge H_2} is to find
formula $H^*_1(W)$ such that \bm{\prob{V}{H_1 \wedge H_2} \equiv H^*_1
  \wedge \prob{X}{H_2}}. Formula $H^*_1$ is called a \tb{solution} to
PQE.
\end{definition}
%
%
\begin{remark}
  \label{rem:noise}
Note that if $H^*_1$ is a solution to the PQE problem above and a
clause $C \in H^*_1$ is implied by $H_2$ \ti{alone}, then $H^*_1
\setminus \s{C}$ is a solution too. So if all clauses of $H^*_1$ are
implied by $H_2$, then an empty set of clauses is a solution too (in
this case, $H^*_1 \equiv 1$).
\end{remark}

Let $N(X,Y,Z)$ be a combinational circuit where $X,Y,Z$ are sets of
input, internal and output variables respectively. Let $N$ consist of
gates $g_1,\dots,g_m$. A formula $F(X,Y,Z)$ specifying the
functionality of $N$ can be built as $G_1 \wedge \dots \wedge G_m$
where $G_i, 1 \leq i \leq m$ is a formula specifying gate
$g_i$. Formula $G_i$ is constructed as a conjunction of clauses
falsified by the incorrect combinations of values assigned to
$G_i$. Then every assignment satisfying $G_i$ corresponds to a
consistent assignment of values to $g_i$ and vice versa.

\begin{example}
\label{exmp:circ_form}
Let $g$ be a 2-input AND gate specified by $v_3 = v_1 \wedge
v_2$. Then formula $G$ is constructed as $C_1 \wedge C_2 \wedge C_3$
where $C_1 = v_1 \vee \overline{v}_3$, $C_2 = v_2 \vee
\overline{v}_3$, $C_3 = \overline{v}_1 \vee \overline{v}_2 \vee
v_3$. Here, the clause $C_1$, for instance, is falsified by assignment
$v_1=0,v_2 = 1,v_3=1$ inconsistent with the truth table of $g$.
\end{example}

\section{Generation Of False Properties}
\label{sec:gen_fls_props}
In this section, we describe generation of false properties of a
combination circuit by PQE.  Subsection~\ref{ssec:motiv} explains our
motivation for building false properties.
Subsection~\ref{ssec:fls_prop} presents construction of false
properties by PQE.  In Subsection~\ref{ssec:test_gen}, we describe
generation of tests breaking these properties and argue that these
tests are of very high quality.

%
%
\subsection{Motivation for building false properties}
\label{ssec:motiv}
Let \abs{P}=\s{P_1(X,Z),\dots,P_k(X,Z)} be a specification of a
combinational circuit. Here $P_i$,$i=1,\dots,k$ are
properties\footnote{The fact that $P_i$ is a property means that an
  implementation satisfying $P_i$ cannot output \pnt{z} for an input
  \pnt{x} if $P_i(\pnt{x},\pnt{z}) = 0$.}  of this circuit and $X$ and
$Z$ are the sets of input and output variables respectively. Let
$N(X,Y,Z)$ be an implementation of the specification \abs{P} where $Y$
is the set of internal variables. Let $F(X,Y,Z)$ be a formula
describing $N$ (see Section~\ref{sec:basic}). Assume that $N$
satisfies all properties of \abs{P} i.e. $F \imp P_i$, $i=1,\dots,k$.

The specification \abs{P} is complete if it fully defines the
input/output behavior of $N$ i.e. if $P_1 \wedge \dots \wedge P_k \imp
\prob{Y}{F}$. Suppose that \abs{P} is \ti{incomplete}. As we mentioned
in the introduction, this may lead to overlooking buggy input/output
behaviors of $N$. In this paper, we address this problem by generating
false properties of $N$. The latter are meant to imitate properties
absent from \abs{P} that $N$ does not satisfy.  The idea here is that
tests breaking these false properties may expose the buggy behaviors
mentioned above.

%
%
\subsection{Building false properties by PQE}
\label{ssec:fls_prop}
Let $F^*(X,Y,Z)$ be a formula obtained from $F$ by replacing a set of
clauses $G$ with those of $G^*$.  Let $F' = F \setminus G$. (So, $F =
G \wedge F'$.) Then the formula $F^*$ equals $G^* \!\wedge\!F'$.  Let
$\tbl(X,Z)$ and $\stbl(X,Z)$ denote the ``truth tables'' of $F$ and
$F^*$ respectively. That is $\tbl\!=\!\prob{Y}{F}$ and
$\stbl(X,Z)\!=\!\prob{Y}{F^*}$. An informal requirement to $G^*$ is
that it is unlikely to be implied by $F$. (Otherwise, the technique we
describe below cannot produce a false property.)  One more
requirement\footnote{If this requirement does not hold, $F^*$ may imply a non-empty set of
clauses $Q(X)$.  Since $F \not\imp Q$, the formula $Q$ is a trivial
false property that just excludes some input assignments to $X$. (So
any input falsifying $Q$ is a counterexample.)  In reality, the
requirement in question can be \ti{ignored} if one also ignores the
spurious false properties $Q(X)$.
} to $G^*$ is that for every
assignment \pnt{x} there exists \pnt{z} such that
$\stbl(\pnt{x},\pnt{z})=1$. (This is trivially true for $\tbl(X,Z)$
because the latter is derived from $F$ specifying a \ti{circuit}.)

Let $Q(X,Z)$ be a solution to the PQE problem of taking $G^*$ out of
the scope of quantifiers in \prob{Y}{G^* \wedge F'}.  That is
$\prob{Y}{G^* \wedge F'} \equiv Q \wedge \prob{Y}{F'}$. The
proposition below shows that $Q$ is a \tb{false property} of $N$ iff
the truth tables \tbl and \stbl are incompatible i.e.  $\tbl \not\imp
\stbl$. (If $\tbl \imp \stbl$ then \stbl can be viewed just as a
relaxation of \tbl).
%
%
\begin{proposition}
\label{prop:fls_prop}
  $F \not\imp Q$ iff $\tbl \not\imp \stbl$.
\end{proposition}
\begin{proof}
\ti{The ``if'' part.}  Let $\tbl \not\imp \stbl$.  Let
(\pnt{x},\pnt{z}) be an assignment to $X \cup Z$ such that
$\tbl(\pnt{x},\pnt{z}) \not\imp \stbl(\pnt{x},\pnt{z})$. (That is
$\tbl(\pnt{x},\pnt{z})\!=\!1$ and $\stbl(\pnt{x},\pnt{z})\!=\!0$.) Let
\pnt{p}=(\pnt{x},\pnt{y},\pnt{z}) be the assignment describing the
execution trace in $N$ under the input \pnt{x}. Then \pnt{p} satisfies
$F$ and hence $F'$.  Assume the contrary, i.e. $F \imp Q$. Then
$Q(\pnt{x},\pnt{z})=1$. Since \pnt{p} satisfies $F'$, then $Q \wedge
\prob{Y}{F'}$ =\prob{Y}{G^* \wedge F'}= \prob{Y}{F^*}=\stbl=1 under
assignment (\pnt{x},\pnt{z}). So we have a contradiction.
  
\ti{The ``only if'' part.}  Let $F \not\imp Q$. Let
\pnt{p}=(\pnt{x},\pnt{y},\pnt{z}) be an assignment to \V{F} that
satisfies $F$ and falsifies $Q$ (i.e. \pnt{p} breaks $F \imp
Q$). Since $F(\pnt{p}) = 1$, then $\tbl(\pnt{x},\pnt{z})=1$. Since
$Q(\pnt{x},\pnt{z})=0$, then $Q \wedge \prob{Y}{F'}$ = \prob{Y}{G^*
  \wedge F'} = \prob{Y}{F^*} = \stbl = 0 under assignment
(\pnt{x},\pnt{z}). So, $\tbl(\pnt{x},\pnt{z}) \not\imp
\stbl(\pnt{x},\pnt{z})$.
 
\end{proof}
\subsection{Test generation}
\label{ssec:test_gen}
Let $Q(X,Z)$ be a property of $N(X,Y,Z)$ obtained as described in the
previous subsection. We are interested in tests that break $Q$.  A
single test of that kind can be extracted from an assignment
\pnt{p}=(\pnt{x},\pnt{y},\pnt{z}) that satisfies $F$ and falsifies $Q$
thus breaking $F \imp Q$. Such an assignment \pnt{p} can be found by
running a SAT-solver on $F \wedge \overline{C}$ where $C$ is a clause
of $Q$.  The \pnt{x} part of \pnt{p} is  \ti{a required test}.

Intuitively, tests breaking $Q$ should be of high quality. The reason
is that they are supposed to break a property that is ``almost
true''. Indeed, the proof of Proposition~\ref{prop:fls_prop} shows
that an assignment (\pnt{x},\pnt{z}) breaking $Q$ exposes the
difference in the input/output behavior specified by $F$ and
$F^*$. The latter is not a trivial task assuming that $F$ and $F^*$
are almost identical. To substantiate our intuition, in the appendix,
we show that stuck-at fault tests (that are tests of a very high
quality) are a special case of tests breaking false properties.


\section{A Complete Set Of False Properties}
\label{sec:cmpl_set}
In the previous section, we introduced false properties as a means to
deal with incompleteness of specification. In this section, we
describe a procedure called \cs that constructs a set of false
properties that is structurally complete. Here we borrow an idea
exploited in testing: if functional completeness is infeasible, run
tests probing every design piece to reach \ti{structural}
completeness.  Similarly, structural completeness of a set of false
properties is achieved by generating properties relating to different
parts of the design.

\subsection{Input/output parameters of the \cs procedure}

 In this section, we continue the notation of the previous section.
 The pseudocode of \cs is shown in Figure~\ref{fig:cmpl_set}.  It
 accepts five parameters: \Abs{P}{hrd},\Abs{P}{inf},$N,F,Y$. The
 parameter \Abs{P}{hrd} is the set of properties of specification
 \abs{P}=\s{P_1,\dots,P_k} that were too hard to prove/disprove. The
 parameter \Abs{P}{inf} is an \ti{informal} specification that is
 assumed to be \ti{complete}\footnote{One can view \Abs{P}{inf} as a
   replacement for the truth table. The role of such a replacement can
   be played, for instance, by the designer.}. The parameter $N$
 denotes a combinational circuit implementing the specification
 \abs{P}.  The parameter $F$ is a formula describing the functionality
 of $N$.  Finally, the parameter $Y$ is the set of internal variables
 of $N$.

%
%
\setlength{\intextsep}{4pt}
\setlength{\textfloatsep}{10pt}
\begin{figure}[h]
\centering
\small
\parbox{0cm}{\begin{tabbing}
aaa\=bb\=cc\= dd\= \kill
$\cs(\Abs{P}{hrd},\Abs{P}{inf},N,F,Y)$\{\\
\tb{\scriptsize{1}}\> $\abs{T} := \emptyset$ \\
\tb{\scriptsize{2}}\> $\Abs{P}{fls} := \emptyset$ \\
\tb{\scriptsize{3}}\> $\mi{Gates}:=\mi{ExtrGates}(N)$    \\
\tb{\scriptsize{4}}\> while $(\mi{Gates} \neq \emptyset)$ \{ \\
\tb{\scriptsize{5}}\Tt  $g := \mi{PickGate}(Gates)$ \\
\tb{\scriptsize{6}}\Tt  $\mi{Gates} := \mi{Gates} \setminus \s{g}$ \\
\tb{\scriptsize{7}}\Tt  $(G,G^*) := \mi{Change}(F,g)$ \\
\tb{\scriptsize{8}}\Tt  $F' :=F \setminus G$ \\
\tb{\scriptsize{9}}\Tt  $Q := \mi{PQE}(G^*,F',Y)$ \\
\tb{\scriptsize{10}}\Tt  $\mi{Tst} := \mi{RunSat}(F,Q)$ \\
\tb{\scriptsize{11}}\Tt  if ($\mi{Tst} = \mi{nil}$) continue \\
\tb{\scriptsize{12}}\Tt  $\Abs{P}{fls} := \Abs{P}{fls} \cup \s{Q}$ \\
\tb{\scriptsize{13}}\Tt  if $(\mi{BreaksProp}(\Abs{P}{hrd},\mi{Tst})$ \\
\tb{\scriptsize{14}}\ttt    return($\mi{Tst},\abs{T},\Abs{P}{fls}$) \\
\tb{\scriptsize{15}}\Tt  if $(\mi{BreaksSpec}(\Abs{P}{inf},\mi{Tst})$ \\
\tb{\scriptsize{16}}\ttt    return($\mi{Tst},\abs{T},\Abs{P}{fls}$) \\
\tb{\scriptsize{17}}\Tt  $\abs{T} := \abs{T} \cup \s{\mi{Tst}}$\} \\
\tb{\scriptsize{18}}\> return($\mi{nil}$,\abs{T},\Abs{P}{fls})\}  \\

\end{tabbing}}
\vspace{-10pt}
\caption{The \cs procedure}
\label{fig:cmpl_set}
\end{figure}

\cs has three output parameters: $\mi{Tst},\abs{T},\Abs{P}{fls}$.  The
parameter $\mi{Tst}$ denotes a test exposing a bug of $N$ (if
any). The parameter \abs{T} consists of tests generated by \cs that
has identified any bug. (These tests may still be of value e.g. for
regression testing). The parameter \Abs{P}{fls} denotes the set of
false properties generated by \cs.

\subsection{The while loop of the \cs procedure}
The main body of \cs consists of a while loop (lines 4-17).  This loop
is controlled by the set \ti{Gates} consisting of gates of
$N$. Originally, \ti{Gates} is set to the set of all gates of $N$
(line 3). \cs starts an iteration of the loop by extracting a gate $g$
of \ti{Gates} (lines 5-6).  Then \cs computes formula $G^*$ that
replaces the clauses of $G$ (describing the gate $g$) in formula $F$
(line 7). After that, \cs calls a PQE solver to find a formula
$Q(X,Z)$ such that $\prob{Y}{G^* \wedge F'} \equiv Q \wedge
\prob{Y}{F'}$ where $F' = F \setminus G$.  The formula $Q$ above
represents a property of $N$ that is supposed to be
false\footnote{Note that any subset of clauses of $Q$ is a property as
  well.  So, to decrease the complexity of PQE-solving, one can stop
  it when a threshold number of clauses is generated. Moreover, from
  the viewpoint of test generation, one can stop PQE \ti{as soon as} a
  clause $C(X,Z)$ not implied by $F$ is generated.}.

\cs checks if $Q$ is indeed a false property by running a SAT-solver
that looks for an assignment breaking $F \imp Q$ (line 10). If this
SAT-solver fails to find such an assignment, $Q$ is a true property.
In this case, \cs starts a new iteration (line 11). Otherwise, the
SAT-solver returns a test \ti{Tst} and $Q$ is added to \Abs{P}{fls} as
a new false property. Then \cs checks if \ti{Tst} breaks an unproved
property of \Abs{P}{hrd}. If it does, then \cs terminates (lines
13-14). After that, \cs checks if \ti{Tst} violates the informal
specification \Abs{P}{inf}. If so, then \cs terminates (line
15-16). Finally, \cs adds \ti{Tst} to \abs{T} and starts a new
iteration.

\section{Extension To Sequential Circuits}
\label{sec:seq_circs}
In this section, we extend our approach to sequential
circuits. Subsection~\ref{ssec:seq_defs} provides some relevant
definitions.  In Subsection~\ref{ssec:hi_lvl}, we give a high-level
view of building a structurally complete set of false properties for a
sequential circuit (in terms of safety properties).  Finally,
Subsection~\ref{ssec:fls_props_seq} describes generation of false
safety properties.

%
%
\subsection{Some relevant definitions}
\label{ssec:seq_defs}
Let $M(S,X,Y,S')$ be a sequential circuit. Here $X,Y$ denote input and
internal combinational variables respectively and $S,S'$ denote the
present and next state variables respectively. Let $F(S,X,Y,S')$ be a
formula describing the circuit $M$. ($F$ is built for $M$ in the same
manner as for a combinational circuit $N$, see
Section~\ref{sec:basic}.)  Let $I(S)$ be a formula specifying the
\ti{initial states} of $M$.  Let $T(S,S')$ denote \Prob{X}{Y}{F}
i.e. the \ti{transition relation} of $M$.

A \ti{state} \pnt{s} is an assignment to $S$. Any formula $P(S)$ is
called a \ti{safety property} for $M$. A state \pnt{s} is called a
$P$-\ti{state} if $P(\pnt{s})=1$. A state \pnt{s} is called
\ti{reachable} \ti{in} $n$ \ti{transitions} (or in $n$-th \ti{time
  frame}) if there is a sequence of states
\ppnt{s}{1},\dots,\ppnt{s}{n+1} such that \ppnt{s}{1} is an $I$-state,
$T(\ppnt{s}{i},\ppnt{s}{i+1})=1$ for $i=1,\dots,n$ and
\ppnt{s}{n+1}=\pnt{s}.

We will denote the \ti{reachability diameter} of $M$ with initial
states $I$ as \di{M,I}. That is if $n=\di{M,I}$, every state of $M$ is
reachable from $I$-states in at most $n$ transitions. We will denote
as \bm{\rch{M,I,n}} a formula specifying the set of states of $M$
reachable from $I$-states in $n$ transitions. We will denote as
\bm{\rch{M,I}} a formula specifying all states of $M$ reachable from
$I$-states.  A property $P$ \ti{holds} for $M$ with initial states
$I$, if no $\overline{P}$-state is reachable from an $I$-state.
Otherwise, there is a sequence of states called a \ti{counterexample}
that reaches a $\overline{P}$-state.

%
%
\subsection{High-level view}
\label{ssec:hi_lvl}
In this paper, we consider a specification of the sequential circuit
$M$ above in terms of safety properties. So, when we say a
specification property $P(S)$ of $M$ we \ti{mean a safety property}.
Let $F_{1,n}$ denote $F_1 \wedge \dots \wedge F_n$ where $F_i$, $1
\leq i \leq n$ is the formula $F$ in $i$-th time frame i.e. expressed
in terms of sets of variables $S_i,X_i,Y_i,S_{i+1}$.  Formula
\rch{M,I,n} can be computed by QE on formula \prob{W_{1,n}}{I_1 \wedge
  F_{1,n}}.  Here $I_1 = I(S_1)$ and $W_{1,n} = \V{F_{1,n}} \setminus
S_{n+1}$. If $n \ge \di{M,I}$, then \rch{M,I,n} is also \rch{M,I}
specifying all states of $M$ reachable from $I$-states.

Let $\abs{P} = \s{P_1,\dots,P_k}$ be a set of properties forming a
\ti{specification} of a sequential circuit with initial states defined
by $I$.  Let a sequential circuit $M$ be an implementation of the
specification \abs{P}.  So every property $P_i$,$i=1,\dots,k$ holds
for $M$ and $I$ i.e. $\rch{M,I} \imp P_i$ , $i=1,\dots,k$. Verifying
the completeness of \abs{P} reduces to checking if $P_1 \wedge \dots
\wedge P_k \imp \rch{M,I}$.  Assume that proving this implication is
hard or it does not hold.  If \abs{P} is incomplete, $M$ may be buggy
for two reasons mentioned in the introduction. In particular, $M$ may
falsify a property absent from \abs{P}. (This means that there is a
reachable state \pnt{s} of $M$ that satisfies all properties of
\abs{P} and \pnt{s} is supposed to be unreachable.)

One can deal with the problem above like it was done in the case of
combinational circuits. Namely, one can use PQE to build false
properties that are supposed to imitate properties that the
specification \abs{P} has missed. By building counterexamples one
generates ``interesting'' reachable states. If one of these reachable
states is supposed to be \ti{unreachable}, $M$ has a bug.

%
%
\subsection{Generation of false properties}
\label{ssec:fls_props_seq}
False properties of a sequential circuit $M(S,X,Y,S')$ can be
generated as follows.  Recall that formula \prob{W_{1,n}}{I_1 \wedge
  F_{1,n}} specifies \rch{M,I,n} (see the previous subsection).  Here
$I_1 = I(S_1)$ and $W_{1,n} = \V{F_{1,n}} \setminus S_{n+1}$. Let $G$
be a (small) subset of clauses of $F_{1,n}$.  Let $G^*$ be a set of
clauses meant to replace $G$.  We assume that $G^*$ satisfies two
requirements similar to those mentioned in
Subsection~\ref{ssec:fls_prop}. (In particular, we impose an informal
requirement that $G^*$ is unlikely to be implied by $I_1 \wedge
F_{1,n}$.)

Let $Q(S_{n+1})$ be a solution to the problem of taking $G^*$ out of
the scope of quantifiers in \prob{W_{1,n}}{I_1 \wedge G^* \wedge
  F'_{1,n}} where $F'_{1,n}=F_{1,n} \setminus G$. That is
\prob{W_{1,n}}{I_1 \wedge G^* \wedge F'_{1,n}} $\equiv Q \wedge $
\prob{W_{1,n}}{I_1 \wedge F'_{1,n}}. By definition, $Q$ is a property
$M$ (as a predicate depending only on variables of $S$). To show that
$Q$ is a \ti{false} property one needs to find an assignment breaking
$I_1 \wedge F_{1,n} \imp Q$. If such an assignment exists, there is a
counterexample proving that a $\overline{Q}$-state is reachable in $n$
transitions. In this case, $Q$ is indeed a false property.

Using the idea above and a procedure similar to that of
Section~\ref{sec:cmpl_set}, one can build a structurally complete set
of false safety properties.

\section{Our Approach And Fault/Mutation Detection}
\label{sec:relation}
Generation of tests breaking false properties is similar to
fault/mutation detection. In manufacturing testing, one generates
tests detecting faults of a predefined set~\cite{abram,atpg}. Often,
these faults (e.g. stuck-at faults) do not simulate real defects but
rather model logical errors.
In software verification, one of old techniques gaining its popularity
is mutation testing~\cite{budd,soft_testing}. The idea here is to
introduce code mutations (e.g. simulating common programmer mistakes)
to check the quality of an existing test suite or to generate new
tests.

Our approach has three potential advantages. First, PQE solving
introduces a new way to generate tests detecting faults/mutations. (In
the appendix, we give an example of using PQE for stuck-at fault
testing.) The appeal of PQE here is that it can take into account
subtle structural properties like unobservability. So PQE-solvers can
potentially have better scalability than tools based purely on
identifying logical inconsistencies (also known as conflicts).

The second advantage of our approach is that it transforms a
fault/mutation into a property i.e. into a \ti{semantic} notion.  This
has numerous benefits. One of them is that a false property specifies
a large number of tests (rather than a single test).  Suppose, for
instance, that we need to find a single test detecting faults $\phi_1$
and $\phi_2$ of a circuit.  An obvious problem here is that a test
detecting one fault may not detect the other.  In our approach,
$\phi_1$ and $\phi_2$ are cast as false properties $Q_1$ and $Q_2$. To
break $Q_1$ or $Q_2$ one needs to come up with a test satisfying
$\overline{Q}_1$ or $\overline{Q}_2$. To break \ti{both} properties
one just needs to find a test satisfying $\overline{Q}_1 \wedge
\overline{Q}_2$ (if any).

Third, our approach can be applied to \ti{abstract} formulas that may
not even describe circuits or programs. So, in a sense, the machinery
of false properties can be viewed as a generalization of
fault/mutation detection.

\section{Conclusions}
\label{sec:concl}
Having an incomplete specification may lead to a buggy implementation.
One of the problems here is that this implementation may not satisfy a
property omitted in the specification. We address this problem by
generating false properties i.e. those that are not consistent with
the implementation. The idea here is that a test breaking a false
property may also expose a bug in the implementation. False properties
are generated by a technique called partial quantifier elimination
(PQE).

Our three conclusions are as follows. First, the machinery of false
properties can be applied to verification of combinational and
sequential circuits. The efficiency of this machinery depends on that
of PQE solving. So developing powerful PQE algorithms is of great
importance. Second, the machinery of false properties can be viewed as
a generalization of fault/mutation detection. On one hand, this
implies that tests breaking false properties are of high quality.  On
the other hand, this means that the machinery of false properties can
be applied to abstract formulas.  Third, by generating properties
whose falsehood is caused by different parts of the design, one
generates a ``structurally complete'' set of false properties. Using
tests that break all properties of this set can significantly increase
the quality of testing.

\bibliographystyle{plain}
\bibliography{short_sat,local}
\vspace{15pt}
\appendix[Stuck-At Fault Tests And False Properties]
In this appendix, we relate stuck-at faults tests and those breaking
false properties built by PQE. This relation suggests that tests
breaking false properties are of high
quality. Subsection~\ref{ssec:recall} briefly recalls stuck-at fault
testing. In Subsection~\ref{ssec:spec_case}, we consider a special
case of Proposition\!~\ref{prop:fls_prop} where a false property is
generated by modifying the original circuit to \ti{another
circuit}. Subsection~\ref{ssec:modeling} describes how stuck-at faults
are modeled in our approach. Finally, in
Subsection~\ref{ssec:st_at_by_pqe}, we discuss generation of stuck-at
fault tests by PQE.
%
%
\subsection{Recalling stuck-at fault testing}
\label{ssec:recall}
A stuck-at fault is an abstract model of a fault in a combinational
circuit where a line is stuck either at value 0 (stuck-at-0 fault) or
1 (stuck-at-1 fault). In current technology, a stuck-at fault does not
simulate an actual defect but rather serves as a \ti{logical} fault
model. Tests detecting stuck-at faults are typically used in
manufacturing testing. However, they can also be employed in design
verification. The appeal of the stuck-at model is in the high-quality
of tests detecting stuck-at faults. It can be attributed to probing
corner input/output behaviors by these tests.

%
%
\subsection{A special case of false properties built by PQE}
\label{ssec:spec_case}
In this section, we continue the notation of
Section~\ref{sec:gen_fls_props}. Let $N(X,Y,Z)$ be a combinational
circuit and $F(X,Y,Z)$ be a formula specifying $N$. Let $G$ be a
subset of clauses of $F$. Then formula $F$ can be represented as
$G \wedge F'$ where $F' = F \setminus G$. Let $F^*$ be a formula
obtained from $F$ by replacing $G$ with a set of clauses $G^*$
i.e. $F^* = G^* \wedge F'$.

In Subsection~\ref{ssec:fls_prop}, we considered generation of
property $Q(X,Z)$ by taking $G^*$ out of the scope of quantifiers in
\prob{Y}{G^* \wedge F'}. There, we imposed the requirement that for
every assignment \pnt{x} there was \pnt{z} such that
\stbl(\pnt{x},\pnt{z})=1 where \stbl = \prob{Y}{F^*}. In this section,
we strengthen this requirement by claiming that for every \pnt{x}
there exists exactly one \pnt{z} such that
\stbl(\pnt{x},\pnt{z})=1. Then one can formulate a stronger version of
Proposition~\ref{prop:fls_prop} (recall that \tbl denotes
\prob{Y}{F}).
%
%
\begin{proposition}
\label{prop:mod_prop}
 $F \not\imp Q$ iff $\tbl \not\equiv \stbl$.
\end{proposition}
\begin{proof}
Since we consider a special case, Proposition~\ref{prop:fls_prop}
holds and so $F \not\imp Q$ iff $\tbl \not\imp \stbl$. Let us show
that under the new requirement to \stbl above, $\tbl \imp \stbl$
entails $\tbl \equiv \stbl$. (Since $\tbl \equiv \stbl$ trivially
implies $\tbl \imp \stbl$, this means that $\tbl \not\equiv \stbl$
entails $\tbl \not\imp \stbl$.)  Assume the contrary i.e. $\tbl
\not\equiv \stbl$. The only possibility here is that there exists an
assignment (\pnt{x},\pnt{z}) to $X \cup Z$ such that
$\tbl(\pnt{x},\pnt{z})=0$ and $\stbl(\pnt{x},\pnt{z})=1$.  Let \Pnt{z}
be the output produced by circuit $N$ for the input \pnt{x}.  Then
\tbl(\pnt{x},\Pnt{z})=1. Since $\tbl \imp \stbl$ then
\stbl(\pnt{x},\Pnt{z})=1 too. But this violates the strengthened
requirement on \stbl because $\stbl(\pnt{x},\pnt{z})$ and
 $\stbl(\pnt{x},\Pnt{z})$ are equal to 1 for the same \pnt{x}.
\end{proof}

\begin{remark}
\label{rem:atpg_test}
The strengthened requirement imposed on \stbl holds if, for instance,
$F^*$ specifies a circuit $N^*$ obtained by a modification of
$N$. Proposition~\ref{prop:mod_prop} implies that a test breaking the
property $Q$ also makes $N$ and $N^*$ produce different outputs and
vice versa. So, if $N^*$ describes a faulty version of $N$, a test
breaking $Q$ detects this fault and vice versa.
\end{remark}
%
%
\subsection{An example of modeling a stuck-at fault}
\label{ssec:modeling}

Let $g$ be a gate of circuit $N$ given in
Example~\ref{exmp:circ_form}.  That is $g$ is a 2-input AND gate
specified by $v_3 = v_1 \wedge v_2$. The functionality of $g$ is
described by $C_1 \wedge C_2 \wedge C_3$ where $C_1 =
v_1 \vee \overline{v}_3$, $C_2 = v_2 \vee
\overline{v}_3$, $C_3 = \overline{v}_1 \vee \overline{v}_2 \vee
v_3$. Here $C_1,C_2,C_3$ are clauses of the formula $F$ specifying
$N$.  

Let $N^*$ denote the circuit obtained from $N$ by introducing the
stuck-at-0 fault at the output of $g$. A formula $F^*$ specifying
$N^*$ is obtained from $F$ by replacing $C_3$ with the clause
$C^*_3= \overline{v}_1 \vee \overline{v}_2 \vee \overline{v}_3$. (It
is not hard to check that by resolving clauses $C_1$,$C_2$ and $C^*_3$
on $v_1$ and $v_2$ one obtains the clause $\overline{v}_3$.) So $F^* =
C^*_3 \wedge F'$ where $F' = F \setminus \s{C_3}$. By taking $C^*_3$
out of the scope of quantifiers in $F^* = C^*_3 \wedge F'$ one obtains
a property $Q(X,Z)$. That is $\prob{Y}{C^*_3 \wedge F'} \equiv
Q \wedge \prob{Y}{F'}$.
Assume that $Q$ is a false property i.e. there exists an assignment
\pnt{p}=(\pnt{x},\pnt{y},\pnt{z}) breaking
$F \imp Q$. Then \pnt{x} makes $N$ and $N^*$ produce different outputs
i.e. \pnt{x} is a \ti{test} detecting the stuck-at fault at hand.

%
%
\subsection{Finding stuck-at fault tests by PQE}
\label{ssec:st_at_by_pqe}
If one needs to find a single test detecting a stuck at-fault, there
is no need to generate the entire false property $Q$ above. For
instance, in the previous subsection, one can stop taking $C^*_3$ out
of the scope quantifiers, as soon as a clause $B(X,Z)$ not implied by
$F$ is generated. Then a test can be extracted from an assignment
satisfying $F \wedge \overline{B}$ (i.e. breaking $F \imp B$). So, PQE
can be used for generation of fault-detecting tests.

Modern tools for generation of fault detecting tests are a combination
of dedicated ATPG methods pioneered by the D-algorithm~\cite{roth} and
SAT-based algorithms~\cite{grasp,chaff}. To make a generic SAT-solver
work in the ATPG setting, some extra work is done. For instance, extra
variables and clauses are added to simulate signal
propagation~\cite{larrabee}. The appeal of ATPG by PQE-solving is that
the latter takes the best of both worlds. On one hand, like a
SAT-solver with conflict clause learning, a PQE-algorithm employs
powerful methods of learning~\cite{certif}. (In reality, the learning
of a PQE-solver is more powerful since in addition to deriving
conflict clauses, a PQE-solver also learns \ti{non-conflict clauses}.)
On the other hand, the machinery of clause redundancy~\cite{fmcad13},
can take into account some subtle structural properties of the circuit
at hand (e.g. observability). In particular, the redundancy based
reasoning of a PQE-solver makes simulating signal propagation quite
effortless and does not require adding new variables and clauses.

\end{document}